\documentclass[copyright,creativecommons]{eptcs}
\providecommand{\event}{QAPL 2011} 
\usepackage[small]{caption}

\usepackage{amsmath,amssymb}

\usepackage{hyperref}
\usepackage{tikz}
\usepackage{pgflibraryshapes}
\usetikzlibrary{arrows,automata}

\newtheorem{theorem}{Theorem}

\newtheorem{corollary}{Corollary}

\newtheorem{definition}{Definition}
\newtheorem{remark}{Remark}
\newtheorem{example}{Example}

\newenvironment{proof}{{\bf ~Proof}.\enspace}%
{$\Box$}
\def\impact#1{\textcolor{black}{#1}}

\def\jo#1{\textcolor{black}{#1}}
\def\joo#1{\textcolor{black}{#1}}

\newcommand{\V}{V}
\newcommand{\X}{X}
\newcommand{\A}{A}
\newcommand{\E}{E}
\newcommand{\cP}{\mathcal{P}}
\newcommand{\inv}{\mathrm{inv}}
\newcommand{\Inv}{\mathrm{Inv}}
\newcommand{\Init}{\mathrm{Init}}
\newcommand{\Act}{\mathrm{Act}}
\newcommand{\Flow}{\mathrm{Flow}}
\newcommand{\Pre}{\mathit{Pre}}
\newcommand{\pre}{\mathit{pre}}
\renewcommand{\Reset}{\mathrm{Set}}  
\newcommand{\prob}{\mathit{prob}}
\newcommand{\Prob}{\mathit{Prob}}
\newcommand{\post}{\mathsf{post}}
\newcommand{\pos}{\mathit{pos}}
\newcommand{\Pos}{\mathit{Pos}}
\renewcommand{\d}{\mathsf{d}}
\newcommand{\Dist}{\mathrm{Dist}}
\newcommand{\supp}{\mathrm{supp}}

\newcommand{\Op}{\mathit{O}}

\newcommand{\OFF}{\mathit{OFF}}

\begin{document}

\title{Analysis of Non-Linear Probabilistic Hybrid Systems}
\def\titlerunning{Analysis of Non-Linear Probabilistic Hybrid Systems}
\def\authorrunning{Assouramou, Desharnais}

\author{
Joseph Assouramou\thanks{Research supported by NSERC}
       \institute{Universit\'e Laval, Quebec, Canada}
       \email{Joseph.Assouramou.1@ulaval.ca}
\and
Jos\'ee Desharnais$^*$
       \institute{Universit\'e Laval, Quebec, Canada}
       \email{Josee.Desharnais@ift.ulaval.ca}
}
\maketitle
\begin{abstract}
This paper shows how to compute, for probabilistic hybrid systems, the clock approximation and linear phase-portrait approximation that have been proposed for non probabilistic processes by
 Henzinger et al.  The techniques permit to define a rectangular probabilistic process from a non rectangular one, hence allowing the model-checking of any class of systems. Clock approximation, which applies under some restrictions, aims at replacing a non rectangular variable   by a clock variable.     Linear phase-approximation applies without restriction and yields an approximation that simulates the original process.
 The conditions that we need for probabilistic processes are the same as those for the classic case.
 \end{abstract}
\section{Introduction}
Hybrid processes are a combination of a process   that evolves continuously with time and of a discrete component.  A typical example is  a physical system, such as a heating unit, that is controlled~by a monitor. There are discrete changes of modes, like turning on and off the unit, and there is a continuous evolution --  the change in temperature.
Because of their continuous nature, model-checking  hybrid systems can only be done for sub-classes of them. Especially, the largest class for which verification is decidable is the class of rectangular hybrid automata~\cite{Isa06,Hen99}. Another such class for which verification is decidable is that of o-minimal hybrid automata~\cite{Lafferriere}, which models hybrid systems whose relevant sets and continuous behavior are definable in an o-minimal structure.  In the probabilistic case, Sproston proposed  methods to verify $\forall$-PBTL  on probabilistic rectangular and o-minimal hybrid processes~\cite{Spr01}. Probabilistic timed automata are also a subclass of such processes and have been analyzed extensively~\cite{Spr04,Spr00}.

In order to allow the verification of non-rectangular hybrid automata, two translation/approximation methods were proposed by Henzinger et al.~\cite{Sim90}: clock-translation and linear phase-portrait approximation. The idea behind those methods is to transfer the verification of any hybrid automaton to the one of a rectangular hybrid automaton which exhibits the same behaviour or over approximates it.
  In this paper, we show how to apply these methods to probabilistic hybrid processes.  We show that both methods apply with the same conditions as for the non deterministic case. The technique of approximation is based on replacing exact values by lower and upper bounds, after splitting the hybrid automaton for more precision in the approximation.  Hence, we also show how to split a  probabilistic hybrid automaton in order to obtain a bisimilar one.   Other side contributions of this paper are: a slightly more general, yet a simpler definition of probabilistic automata than the one proposed by Sproston~\cite{Spr01}; and the  description, in the next background section, of the two translation techniques in a simpler way than  what can be found in~\cite{Sim90,Sim92}, mostly because we take advantage of the fact that the definition of hybrid automata has been simplified since then,  \jo{being presented in terms of functions instead of predicates, and being slightly less general than in the original paper~\cite{Sim92}.}

\section{Transformation methods for hybrid automata}
In this section, we describe the two methods presented by Henzinger et al~\cite{Sim90} that will permit  the verification of safety properties on any hybrid system.

  \label{s:def_AH}
     Let $X=\{x_1,\ldots,x_n\}$ be a finite set of real variables; we write $\dot{X}=\{\dot{x}_1, \ldots, \dot{x}_n\}$ where $\dot{x}_i=\frac{dx_i}{dt}$ is the first derivative of $x_i$ with respect to time.
     The set of predicates on $\dot{\X}\cup \X$ is denoted $G(\dot{\X} \cup \X)$.  \jo{The set of \emph{valuations} $\mathsf{a}:X\to \mathbb{R}$ is written $\mathbb{R}^X$ or $\mathbb{R}^n$. 
      A set $U\subseteq \mathbb{R}^X$ is \emph{rectangular} if there exists a family of (possibly unbounded) intervals $(I_x)_{x\in X}$ with rational endpoints  such that $U=\{\mathsf{a}\in \mathbb{R}^n\mid \mathsf{a}(x)\in    I_x \mbox{ for all } x\in X\}$.  We denote by $R(X)$ the set of rectangles over $X$.  For any set $Y$, we write $\cP(Y)$  (resp.\ $\cP_{\mathit{fin}}(Y)$) for the power set of $Y$ (resp.\  finite power set of $Y$).
} For any variable $x$, belonging to $X$ or not, we write  $\mathsf{a}[x\mapsto r]$ for the valuation that maps $x$ to
$r\in \mathbb{R}$  and agrees with $\mathsf{a}$ elsewhere.  Conversely, if $X'\subseteq X$, we write $\mathsf{a}|_{X'}$ for the restriction of $\mathsf{a}$ to $X'$. In the following, we use the notation $\Reset$ instead of the usual slightly misleading one: $\mathrm{Reset}$.

\begin{definition}\cite{alur2000}\label{defHA}
 $H=(\V, \X, \Init, \Act, \Inv, \Flow,\E$, $\Pre, \Reset)$   is    a \emph{hybrid automaton (HA)} if
\begin{itemize}\addtolength{\itemsep}{-4pt}
\item $\V$ is a finite set of locations or control modes;
\item   $\X=\{x_1,x_2,\ldots, x_n\}$ is a set of $n$ continuous variables;
\item   $\Inv:{\V} \rightarrow{\cP({\mathbb{R}^{X}})}$ defines invariants for the variables in each location.
\item  $\Init: V \rightarrow\cP({\mathbb{R}^{X}})$ defines initial states  \jo{and satisfies $\Init(v) \subseteq \Inv(v)$ for all $v\in V$}.
\item   $\Act$ is a finite set of actions, possibly including a silent one, $\tau$;
\item   $\Flow:\!\V\!\rightarrow{G(\dot{\X}\! \cup \!\X)}$
    is a flow evolution condition;
\item   $\E \subseteq \V\times \Act \times \V$  is a finite set of discrete transitions;
\item   $\Pre: \E  \rightarrow\cP({\mathbb{R}^{X}})$  maps to every discrete transition  a set of preconditions;
\item   $\Reset: \E \times {\mathbb{R}^{X}} \rightarrow\cP({\mathbb{R}^{X}})$ describes change in values of variables resulting from taking  edges.    
We write $\Reset^x(e):= \{\mathsf{d}(x)\mid\exists \mathsf{a}\in \mathbb{R}^X.\mathsf{d}\in\Reset(e,\mathsf{a})\}$.

\end{itemize}

  \jo{   $H$ is said to be \emph{rectangular} if the  image of $\Inv$, $\Pre$ and $\Reset$ are included in $R(X)$ and
$\Flow(v)=\displaystyle\wedge_{x\in X}\ \dot x\in I_x$ where each $I_x \subseteq \mathbb{R}$ is a (possibly unbounded) interval with rational endpoints.}
\end{definition}


The semantics of $H$ is a labelled transition system: the set of states is $S_H:=\{(v,\mathsf{a})\mid\mathsf{a}\in \Inv(v)\}$.
There~are two kinds of transitions between states: flow transitions and discrete transitions. In a flow transition, the mode of the automaton is fixed and only the variables' values change over time. More formally, there is  a flow transition of duration  $\sigma \in \mathbb{R}_{\geq0}$ between states $(v,\mathsf{a})$ and $(v,\mathsf{a}')$,  written  $(v,\mathsf{a})\overset{\sigma}\rightarrow{(v,\mathsf{a}')}$, if either
 (1) $\sigma=0$ and $\mathsf{a}=\mathsf{a}'$ or (2) $\sigma>0$ and there exists a differentiable function $\gamma: [0;\sigma]\rightarrow{\mathbb{R}^n}$ with $\dot{\gamma}: (0;\sigma)\rightarrow{\mathbb{R}^n}$ such that $\gamma$ is a solution of $\Flow(v)$ with $\gamma(0)=\mathsf{a}$, $\gamma(\sigma)=\mathsf{a}'$, and $\gamma(\epsilon) \in \textrm{Inv}(v)$ for all $\epsilon \in (0;\sigma)$.
 For discrete transitions, the control mode of the automaton changes instantaneously.
 We write $(v,\mathsf{a}) \overset{a}\rightarrow{(v',\mathsf{a}')}$, if there exists $e=(v,a,v') \in E$ such that $\mathsf{a} \in \Pre(e)$ 
 and $\mathsf{a}' \in \Reset(e,\mathsf{a})$.

%

   \begin{example}
   Figure~\ref{therm} shows the graphical representation of a thermostat%
~\cite{Sim92} that  controls the variation of  temperature in a room through a radiator. 
The whole system has three modes representing that the radiator is either on,  off, or down;  \jo{there is one initial state, where the radiator is on and the temperature has value 2}. When the radiator is on the temperature increases with respect to the equation $\dot{x}=-x + 5$ whereas it decreases with respect to  $\dot{x}=-x$ when the radiator is off. When the whole system is down, no variation of the temperature is modeled. The values of the temperature evolve in the {range $[1;3]$.}
\impact{The radiator must switch off when it is on and the temperature reaches $3$ units and on  when it is off and the temperature is $1$}. Finally,  when we try to turn on the radiator, it might turn on or down.

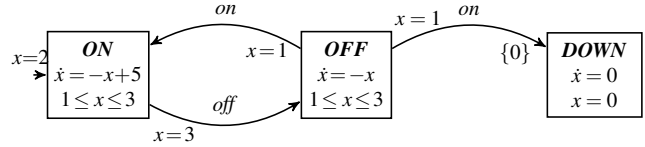
\begin{figure}
\begin{center}
  \begin{tikzpicture}[->,>=stealth',shorten >=0.3pt,auto,node distance=3.3cm,
                    semithick]
  \tikzstyle{every state}=[rectangle,fill=none,draw=black,text=black, minimum width=0em]
\node[state, draw=none] (R) { };
  \node[state,  node distance=1cm] (A)  [right of=R]
  {\begin{minipage}{1.09cm}  \begin{center}  \scriptsize{\emph{\textbf{ON}}\\  ${\!\dot{x}=\!-x\!+\!5\!}$\\ ${1 \!\leq x\! \leq\!3}$} \end{center}
\end{minipage}};

  \node[state]         (B) [right  of=A]      {\begin{minipage}{0.94cm}\begin{center}  \scriptsize{\emph{\textbf{OFF}}\\${\dot{x}\!=\!-x}$\\ ${\!1 \!\leq x\! \leq\!  3}$}
\end{center}\end{minipage}};

  \node[state]         (C) [right  of=B]      {\begin{minipage}{1cm}\begin{center}  \scriptsize{\emph{\textbf{DOWN}}\\$\dot{x}=0$\\ $x=0$}
\end{center}\end{minipage}};

  \path   (R) edge  node[pos=-0.1,above] {\scriptsize{$\!x\!\!=\!\!2$}}(A)
  (A) edge     [bend right]     node[pos=0.17, below ] {\scriptsize{$\!x\!=\!3$}}  node[pos=0.9, below , sloped ] {
  }    node[pos=0.5, above , sloped ] {\scriptsize{\emph{off}}} node[pos=0.5, below , sloped ] {}
  (B)
  (B) edge     [bend right]     node[pos=0.22, below  ] {\scriptsize{$\!x\!=\!1\!$}}  node[pos=0.9, below , sloped ] {
  }
      node[pos=0.5, above , sloped ] {\scriptsize{\emph{on}}} node[pos=0.5, below , sloped ] {}
      (A)
  (B)  edge[in=160,out=0, bend left]     node[pos=0.17, above  ] {\scriptsize{$x=1$}}         node[pos=0.5, below , sloped ] {} node[pos=0.8, below  ] {\scriptsize{$\{0\}$}} node[pos=0.5, above , sloped ]{\scriptsize{\emph{on}}}
   (C);
     \end{tikzpicture}
\end{center}
\caption{A graphical representation of the thermostat hybrid automaton}\label{therm} 
\end{figure}

   \end{example}

Because we need a notion of weak simulation, we define weak transitions through stuttering.   \jo{Hence, let $\tau$ be the (usual) silent action.  We write} $s \overset{a}\twoheadrightarrow{s'}$ if there exists a finite sequence  $s\overset{\tau}\rightarrow{s_1}\overset{\tau}\rightarrow\ldots\overset{\tau}\rightarrow{s_k}\overset{a}\rightarrow{s'}$. Similarly, we write $s\overset{\sigma}\twoheadrightarrow{s'}$ if there exists a finite sequence  $s\overset{\sigma_1}\rightarrow{s_1}\overset{\tau}\rightarrow{s_2}\overset{\sigma_2}\rightarrow{}\ldots\overset{\tau}\rightarrow{s_{k}}\overset{\sigma_{k}}\rightarrow{s'}$ such that  $\sum_i{\sigma_i}\!=\!\sigma\in\mathbb{R}_{\geq 0}$.
\emph{Simulation} and \emph{bisimulation}  are defined on the underlying infinite transition system (and are rather called \emph{time bi/simulation in Henzinger et al.~\cite{Sim90}}).

   \begin{definition} \cite{Sim90}
     \label{def_relation_temps_simulation}
Let $H$ and $H'$ be two hybrid automata. A relation $ \preceq \subseteq S_{H}\times S_{H'}$ is a \emph{simulation} of $H$ by $H'$  \jo{if every initial state  of $H$ is related by $\preceq$ to an initial state of   $H'$
and if whenever $s \preceq s'$, then for each $a\in \Act\setminus\{\tau\}\cup\mathbb{R}_{\geq 0}$ and each transition 
     $s \overset{a}\twoheadrightarrow{s_1}$, there exists a transition $s' \overset{a} \twoheadrightarrow{{s'}_1}$ such that ${s}_1 \preceq {{s}'_1}$}. If $\preceq^{-1}$  is also a simulation, then $\preceq$ is called a bisimulation. If there is a simulation between  $H$ and $H'$ (resp.\  a bisimulation), we write $H \preceq H'$ (resp.\ $H \equiv H'$).
           \end{definition}

\subsection{Clock-translation}\label{sec:clock}
  The \emph{clock-translation} method is based on the substitution of non-rectangular variables by clocks.
  Let $H$ be a non-rectangular hybrid automaton. The substitution of a variable $x$ of $H$ by a clock $t_x$ is possible only if, at any time,
   the value of $t_x$ can be determined by the one of $x$ (i.e., $x$ is solvable).

\subsubsection{Preliminaries}
 We say that a predicate is \emph{simple} if it is a  positive boolean combination of predicates of  the form $x\sim c$  where $c \in \mathbb{R}$ and $\sim \in \{<,\leq,=,\geq,>\}$.
 We say that     $x$ is  \emph{solvable} in $H$ if
\begin{itemize}
\item  every initial condition, invariant condition, and precondition  of $H$ defines {simple} predicates for $x$
and each flow condition of $x$ in $\Flow(v)$   
has the form $(\dot x = f^v_x(x))\wedge P_x$,
where $P_x$ is a simple predicate on $x$;
flow evolutions of other variables must not depend on $x$ nor $\dot x$;
\item   
the initial-value problem $\dot y (t)=f^v_x(y(t))$; $y(0)=c$ has a  unique, continuous and strictly monotone solution   $g_c$;

\item 
 $H$ is initialised with respect to $x$.   That is, for any transition $e\in E$,  $x$ must either stay unchanged in any valuation or get  \jo{assigned} only one value $r$ for all valuations;  this will happen if $\Reset^x(e)$ is a singleton  \jo{with the help of the following notation}:
 we will write
 $$\Reset^x(e)=\{r\}\subseteq {\mathbb{R}_*\!:=\! \mathbb{R}\cup \{*\}},$$  where $r$ is either the unique value $r \in \mathbb{R}$, in which case we say that $x$ is reset to $r$ by $e$, or a  special character, $*$, which will represent stability in the value of $x$.  If $r=*$ we must also have that  $f^v_x=f^{v'}_x$. 

\end{itemize}


\begin{example}
The thermostat automaton  of Figure~\ref{therm} is \emph{solvable} as the flow evolution equation of  variable $x$ is solvable in all the modes: in mode ON, the differential equation $\dot{x}=-x+5$ with the initial condition $x(0)=2$ has the function $x(t)=-3e^{-t}+5$ where $t \in \mathbb{R}_{+}$ as solution.
\end{example}

  Suppose that $x \in \X$ is solvable in the hybrid automaton $H=(\V, \X, \Init, \Act, \Inv, \Flow, \E, \Pre, \Reset)$, and let $c \in \mathbb{R}$ be a constant. We say that $c$ is a  \emph{starting value} for a variable $x$ if 
  $c$ is either:
the initial value of $x$ in some mode $v$, that is, $c=\mathsf{a}(x)$ for $\mathsf{a}\in \Init(v)$; or
the unique value of $\Reset^x(e)$ for some edge $e\in E$ if this value is not $*$.
   Let $D_v(x)$ 
   be the finite set of starting values of  $x$ in $v$.  
  \paragraph{Transformation from $x\sim l$ to $t_x \sim' g^{-1}_c(l)$.}
  To simplify the presentation below, we show how predicates on $x$ are transformed into predicates on $t_x$~\cite{Sim90}.

        Let $g_c(t)$ be the unique solution of the initial-value problem  $\dot y (t)=f^v_x(y(t)); y(0)=c$, where $c \in \mathbb{R}$. As $g_c(t)$ is strictly monotone, there exists at most one $t \in \mathbb{R}_{+}$ such that $g_c(t)=l$ for each $l \in \mathbb{R}$. Let $g_c^{-1}(l)=t$ if $g_c(t)=l$ and $g_c^{-1}(l)=-$ if $g_c(t)\neq l$ for all $t \in \mathbb{R}_{+}$. Let $\Op:=\{<,\leq,=,\geq, >\}$. The transformation from simple atomic predicates over $\{x\}$ to simple atomic predicates over $\{t_x\}$ is the function $\alpha_c$ defined using $\sim \in \Op$, $lt:\Op \rightarrow \Op$ and $gt:\Op \rightarrow \Op$, as follows:\\[5pt]
\phantom{mmm}
\begin{minipage}{10cm}
  ${ \alpha_c(x\sim l)\!=\!\begin{cases} \mathit{true}& \text{if}\ g_c^{-1}(l)=- \ \text{and}\ c\sim l.\\
  \mathit{false}& \text{if}\ g_c^{-1}(l)=- \ \text{and}\ c\nsim l.\\
  t_x \ lt(\sim)\  g_c^{-1}(l)& \text{if}\ g_c^{-1}(l)\neq- \ \text{and}\ c\sim l.\\
   t_x\  gt(\sim)\  g_c^{-1}(l)& \text{if}\ g_c^{-1}(l)\neq- \ \text{and}\ c\nsim l.
   \end{cases}}$
\end{minipage}
\begin{minipage}{4cm}
\begin{tabular}{|c|c|c|}
  \hline
  \textbf{$\sim$} & \textbf{$lt(\sim)$} & \textbf{$gt(\sim)$} \\\hline
  $<$ & $<$ & $>$ \\\hline
  $\leq$ & $\leq$ & $\geq$ \\\hline
  $=$ & $=$ &  $=$\\\hline
  $\geq$ & $\leq$ & $\geq$ \\\hline
   $>$& $<$ & $>$ \\
  \hline
\end{tabular}
\end{minipage}

For each $(v,c_i)$ of the hybrid automaton, every predicate $x \sim l$ is replaced by the predicate $\alpha_{c_i}(x\sim l)$, except the invariant predicate which is replaced by $\alpha_{c_i}(x\sim l)$ if $c_i \sim l$, and by \emph{false} otherwise($(v, c_i)$ may be removed in the latter case).

\subsubsection{Clock-translation}
We are now ready to define clock-translation.
\begin{definition}\cite{Sim90}
   If $x \in \X$ is solvable in  ${H=(\V, \X,\Init, \Act,\Inv, \Flow, \E, \Pre, \Reset)},$ then the  clock-translation of $H$  with respect to $x$ is
$$T\!=\!(\V_T, \X_T, \Init_T, \Act, \Inv_T, \Flow_T, \E_T, \Pre_T, \Reset_T),$$ the hybrid system
  obtained from the following algorithm:\\[5pt]
     Step 1: adding the clock $t_x$.
\begin{itemize}\addtolength{\itemsep}{-4pt}
\item $\V_T:=\cup_{v\in V} \{v\}\times D_v(x)$, that is, each mode $v$ of $H$ is split.  $\X_T:= \X\cup\{t_x\}$.
            \item 
               \jo{$\Init_T(v,c):=\{\mathsf{a}[t_x\mapsto 0]\mid \mathsf{a}\in \Init(v) \mbox{ and  } \mathsf{a}(x)=c \}$.}
            \item $\E_T$ contains two kinds of control switches; for ${c \in D_v(x)}$ and $e=(v,a,v')\in \E$
\\[5pt]--  if \ $\Reset^x(e)=\{r\} \subseteq \mathbb{R}$,
$E_T$  contains the edge ${e_T\!:=\!((v,c),a, (v',r))}$, with $\Pre_T(e_T)\!:=\!\Pre(e)$
and  ${\Reset_T(e_T,\mathsf{a}):=\{ \mathsf{d}[t_x\mapsto 0]\mid \mathsf{d}\in \Reset(e,\mathsf{a}|_X)\}}$.
\\[5pt]--   if \   ${\Reset^x(e)=\{*\}}$, $E_T$ contains the edge $e_T:=((v,c),a, (v',c))$ with $\Pre_T(e_T):=\Pre(e)$ and $\Reset_T(e_T,\mathsf{a}):=\{\mathsf{d}[t_x\!\mapsto \!\mathsf{a}(t_x)]\!\mid \! \mathsf{d}\!\in \!\Reset(e,\mathsf{a}|_X)\}$.
            \end{itemize}

\noindent     Step 2: {moving to conditions on $t_x$.
We view the images of $\Init$, $\Inv$, $\Pre, \Reset$ and $\Flow$  as predicates instead of sets of valuations.  We replace these predicates  over $x$ of the form $x \sim r$ where $ \sim \,\in \!\{<,\leq, =, >,\geq\}$ and $r \in \mathbb{R}$ in $T$ by predicates over $t_x$ of the form $t_x \sim' g^{-1}_c(r)$ as described above.  Finally, the variable $x$ can be removed from $\X_T$.
      }
\end{definition}

 \begin{example}
   The timed automaton of  Figure~\ref{aut_Thermostat_CT} is  obtained by applying the clock-translation on the thermostat automaton. Each mode $v$ of the automaton is split into $|D_v(x)|$ modes. Since ${D_{\mathit{ON}}(x)=\{1, 2\}}$, we get the  modes $(\mathit{ON},\,1)$ and $(\mathit{ON},2)$. For these two modes, the differential equations are respectively "$\dot{x}\!\!=\!\!-x+5$,\ where\ $x(0)\!\!=\!\!2$" and "$\dot{x}\!\!=\!\!-x+5$, where $x(0)\!\!=\!\!1$", and then we have the solutions "$x(t)\!\!=\!\!-3e^{-t}+5$" and "$x(t)\!\!=\!\!-4e^{-t}+5$" respectively. Next, we substitute the variable $x$ by the clock $t_x$ in the preconditions, and the invariants. Then, the constraint $x\leq 3$ becomes  $t\leq \ln(2)$ and $t\leq \ln(\frac{3}{2})$ respectively.
   The two automata are bisimilar, by the following theorem.

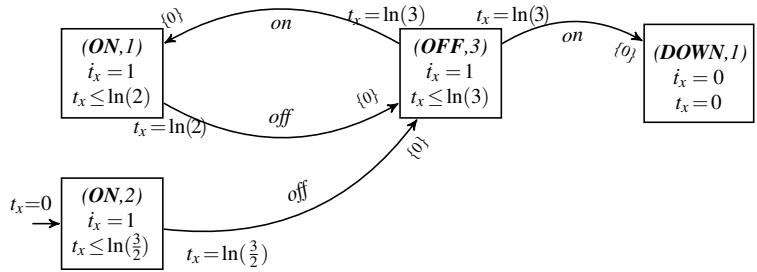
\begin{figure}
\begin{center}
  \begin{tikzpicture}[->,>=stealth',shorten >=0.3pt,auto,node distance=3.3cm,
                    semithick]
  \tikzstyle{every state}=[rectangle,fill=none,draw=black,text=black, minimum width=0em]

  \node[state,  node distance=0.3cm] (A)                    {\begin{minipage}{1.09cm}  \begin{center} \scriptsize{\emph{(\textbf{ON},1)}\\  $\dot{t}_x=\!1$\\ ${t_x\! \leq \!\ln(2)}$} \end{center}
\end{minipage}};

 \node[state, node distance=2cm] (D) [below of=A]               {\begin{minipage}{1.09cm}  \begin{center}  \scriptsize{\emph{(\textbf{ON},2)}\\  $\dot{t}_x=\!1$\\ ${\!t_x\! \leq \! \ln(\!\frac{3}{2}\!)}$} \end{center}
\end{minipage}};
  \node[state, draw=none, node distance = 1.2cm] (R) [left of=D] { };
  \node[state, node distance=4.5cm]         (B) [right  of=A]      {\begin{minipage}{1.09cm}\begin{center}  \scriptsize{\emph{(\textbf{OFF},3)}\\$\dot{t}_x=\!1$\\ ${t_x\! \leq\! \ln(3)}$}
\end{center}\end{minipage}};

  \node[state]         (C) [right  of=B]      {\begin{minipage}{1.2cm}\begin{center}  \scriptsize{\emph{(\textbf{DOWN},1)}\\$\dot{t}_x=0$\\ $t_x\!=\!0$}
\end{center}\end{minipage}};

  \path   (R) edge  node[pos=0,above] {\scriptsize{$\!t_x\!\!=\!\!0$}}(D)
  (A) edge     [bend right]     node[pos=0.03, below] {\scriptsize{$\!t_x\!=\!\ln(\!2 \!)\!$}}
    node[pos=0.5, above , sloped ] {\scriptsize{\emph{off}}} node[pos=0.5, below ] {}node[pos=0.9, above , sloped ] {\tiny{$\{0\}$}}(B)

  (D) edge     [bend right]     node[pos=0.21, below ] {\scriptsize{$t_x\!=\! \ln(\frac{3}{2})$}}
  node[pos=0.5, above , sloped ] {\scriptsize{\emph{off}}} node[pos=0.5, below , sloped ] {}node[pos=0.93, below , sloped ] {\tiny{$\{0\}$}}(B)

  (B) edge     [bend right]     node[pos=0.05, above ] {\scriptsize{$t_x\!=\!\ln(3)$}}  node[pos=0.93, below , sloped ] {
  }
  node[pos=0.5, below , sloped ] {\scriptsize{\emph{on}}} node[pos=0.5, below , sloped ] {}node[pos=0.93, above , sloped ] {\tiny{$\{0\}$}}(A)

  (B)  edge[in=160,out=0, bend left]     node[pos=0.09, above ] {\scriptsize{$\!t_x\!=\!\ln(\!3\!)$}}         node[pos=0.5, below , sloped ] {} node[pos=0.93, below , sloped ] {\tiny{$\{0\}$}} node[pos=0.5, below , sloped ]{\scriptsize{\emph{on}}}
  (C);
     \end{tikzpicture}
\end{center}

    \caption{The clock-translation of the thermostat automaton}\label{aut_Thermostat_CT}
\end{figure}

\end{example}


\begin{theorem}\cite{Sim90}
   \label{theorem_bis_trans}
$H$ is bisimilar to its clock-translation  $T$.  The relation is given by the graph of the projection $\eta:S_T \rightarrow{S_H}$, defined as $\eta((v,c),\mathsf{a})\!:=\!(v,\mathsf{a}')$, where   $\mathsf{a}'$ satisfies  $\mathsf{a}'|_X=\mathsf{a}|_X$ and  $\mathsf{a}'(x)\!\!=\!\!g_c(\mathsf{a}(t_x))$ where $g_c$ is the solution of the initial-value problem $[{\dot y (t)=f^v_x(y(t)); y(0)=c}]$.
   \end{theorem}

  As a corollary, $H$ and $T$ satisfy the  \jo{same properties of usual temporal logics}.


\subsection{Linear phase-portrait approximation}

  \label{approximation}
We now present the second method which allows the translation of any hybrid automaton into a rectangular one.
The linear phase-portrait approximation method can be applied to any hybrid automaton, yielding an \emph{approximation} of the original process which   simulates  the original automaton (instead of being bisimilar to it, as for clock-translation).
This implies  that if a safety property is verified on  the approximation, then it holds in the original system~\cite{Sim92}.

  The general method  is to first split the automaton and then approximate the result.  Approximation is done by replacing non-rectangular flow equations by lower and upper bounds on the variables, hence forgetting the true details of the equations.    
  By splitting more finely, one obtains a better approximation.

\subsubsection {Splitting a hybrid automaton}
 Let $H$  be a hybrid automaton with invariant function $\Inv$. 
A \emph{split function} is a map $\theta$  that returns to each mode $v$ of $H$ a finite open cover $\{\inv^{v}_1,\ldots,\inv^{v}_m\} \subseteq\cP({\mathbb{R}^{X}})$  of $\textrm{Inv}(v)$.   In splitting, a mode $v$ will be split into several modes according to the cover $\theta(v)$. The fact that  $\cup_{i}\inv^{v}_i=\Inv(v)$ makes sure that states are preserved whereas the evolution inside mode $v$ is preserved  through silent transitions between copies of $v$, which is possible because   $\theta(v)$'s components overlap.

\begin{definition}
 Let $H=(\V, \X, \Init, \Act, \Inv, \Flow, \E, \Pre, \Reset)$  be a hybrid automaton.   The split of $H$  by
  $\theta$ is the hybrid automaton
  ${\theta(H)=(\V_{\theta}, \X, \Init_{\theta}, \Act_{\theta},} $${\Inv_{\theta}, \Flow_{\theta}, \E_{\theta}, \Pre_{\theta}, \Reset_{\theta})}$ defined as:\vspace{-1mm}

   \begin{itemize}\addtolength{\itemsep}{-4pt}
   \item $\V_{\theta}= \{(v,i)\mid v \in \V\mbox{ and } 1\leq i \leq |\theta(v)|\}$
   \item $\Init_{\theta}((v,i))=\Init(v) \cap \inv^v_i$ 
   \item $\Act_{\theta}=\Act \cup \{\tau\}$
   \item $\Inv_{\theta}(v,i)=\inv^{v}_i$ 
   \item $\Flow_{\theta}(v,i)=\Flow(v)$
   \item ${\E_{\theta}=\E_1\cup \E_2}$, where $\E_1$ contains the control switch $((v,i),a,(v',j))$ for each $(v,a,v') \in \E$, whereas
     ${\E_2=\{((v,i),\tau,(v,j))\mid (v,i),(v,j)\in \V_{\theta} \}}$ allows the automaton to transit
     silently between the different copies of $v$.
   \item If ${e_{\theta}\!=\!((v,i),a,(v',j)) \in E_1}$, we set
   $\Pre_{\theta}(e_{\theta})\!\!=\!\!\Pre(v,a,v')$ and ${\Reset_{\theta}(e_{\theta},\mathsf{a})\!=\!\Reset((v,a,v'),\mathsf{a})}$.  If $e_{\theta}\!=\!((v,i),\tau,(v,j)) \in \E_2$, we set $\Pre_{\theta}(e_{\theta})=\mathbb{R}^X$ and $\Reset_{\theta}(e_{\theta},\mathsf{a})=\{\mathsf{a}\}$.
   \end{itemize}
\end{definition}

Note that the cover $\theta(v)$ need not really be open.  What is important is that the evolution within any mode be preserved, as pointed out in~\cite{Sim90}. This is the case in the following example, where components of the cover are closed and intersect in exactly one point,  which is sufficient to allow evolution.
  \begin{example}
   The automaton in Figure~\ref{aut_Thermostat_split} is a split of the thermostat automaton with function  $\theta(\mathit{ON})=\theta(\mathit{OFF})=\{{\{x\mid 1\leq x \leq2\}},$ ${ \{x\mid 2\leq x \leq 3\}\}}$, and   $\theta(\mathit{DOWN})\!=\!\{\{\!x\!\mid \!x\!=\!0\!\}\}$.     \joo{Note the silent transitions between states $((\mathit{ON},i),2)$, $i=1,2$, the latter being duplicates of the original thermostat's state $(\mathit{ON},\{x\mapsto 2\})$, that preserve the evolution within mode $\mathit{ON}$.}
  \end{example}
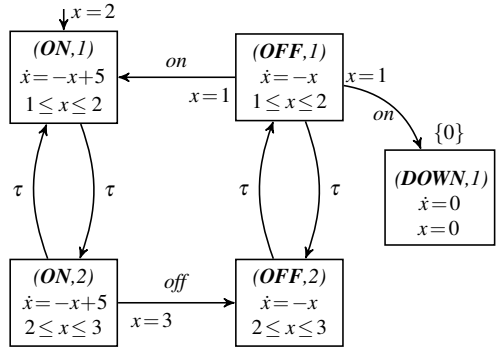
\begin{figure}
\centering
\begin{tikzpicture}[->,>=stealth',shorten >=0.1pt,auto,node distance=3cm,                    semithick]
  \tikzstyle{every state}=[rectangle, fill=none,draw=black,text=black, minimum width=3.7em]
  \node[state, draw=none] (R) {};
  \node[state,  node distance=14mm] (A)  [below of=R]  {\begin{minipage}{1.08cm}  \begin{center}\scriptsize{\emph{(\textbf{ON},1)}\\[1pt]  ${\!\dot{x}\!=\!-x\!+\!5}$\\[1pt]  ${1 \!\leq x\! \leq  2}$}\end{center}
\end{minipage}};

 \node[state, node distance=3cm] (D)  [below of=A]
 {\begin{minipage}{1.08cm}  \begin{center} \scriptsize{\emph{(\textbf{ON},2)}\\  ${\dot{x}\!=\!-x\!+\!5}$\\ ${2 \!\leq x\! \leq  3}$}\end{center}
\end{minipage}};

  \node[state]         (B) [right  of=A]      {\begin{minipage}{1.08cm}\begin{center}  \scriptsize{\emph{(\textbf{OFF},1)}\\${\dot{x}\!=\!-x}$\\ ${1 \!\leq x\! \leq\!  2}$}
\end{center}\end{minipage}};

 \node[state, node distance=3cm]         (E) [below  of=B]      {\begin{minipage}{1.08cm}\begin{center}  \scriptsize{\emph{(\textbf{OFF},2)}\\$\dot{x}\!=\!-x$\\ ${2 \!\leq x\! \leq\!  3}$}
\end{center}\end{minipage}};

  \node at (5,-3)[state]         (C)       {\begin{minipage}{1.2cm}\begin{center}  \scriptsize{\emph{(\textbf{DOWN},1)}\\$\dot{x}\!=\!0$\\ $\!x\!=\!0$}
\end{center}\end{minipage}};

  \path
  (R) edge    node[pos=0.1,right] {\scriptsize{$\!x=\!2$}}(A)
  (D) edge     [right]
  node[pos=0.3, below  ] {\scriptsize{$\!x\!=\!3$}}
  node[pos=0.9, below  ] { 
  }
  node[pos=0.5, above  ] {\scriptsize{\emph{off}}}
  node[pos=0.5, below  ] {}(E)
  (B) edge     [ right]
  node[pos=0.2, below  ] {\scriptsize{$\!x\!=\!1$}}
  node[pos=0.9, below  ] {  
  }
   node[pos=0.5, above  ] {\scriptsize{\emph{on}}}
   node[pos=0.5, below  ] {}(A)
  (B)  edge[in=160,out=0, bend left]
  node[pos=0.23, above ] {\scriptsize{$\!x\!=\!1$}}
     node[pos=1.05, above right] {\scriptsize{$\{0\}$}} node[pos=0.4, below  ]{\scriptsize{\emph{on}}}(C)

   (D) edge    [in=250,out=110]    
     node[pos=0.9, below , left ] {  
  }
      node[pos=0.5, above , left ] {\scriptsize{\emph{$\tau$}}}  node[pos=0.5, below , left ] {}(A)
   (A) edge      [in=70,out=290]    
     node[pos=0.9, below , right ] {  
  }
       node[pos=0.5, above,right  ] {\scriptsize{\emph{$\tau$}}}  node[pos=0.5, below  ] {}(D)

   (E) edge    [in=250,out=110]     
       node[pos=0.9, below , left ] {  
  }
      node[pos=0.5, above , left ] {\scriptsize{\emph{$\tau$}}} node[pos=0.5, below  ] {}(B)

    (B) edge     [in=70,out=290]      
       node[pos=0.9, below  ] {  
  }
      node[pos=0.5, above,right  ] {\scriptsize{\emph{$\tau$}}} node[pos=0.5, below ] {}(E)
  ;
 \end{tikzpicture}
    \caption{ A split  of the thermostat} \label{aut_Thermostat_split}
\end{figure}


%
%
%



\subsubsection{Approximating a hybrid automaton}
An (over) approximation of a HA is obtained by weakening all predicates of its evolution.
  \begin{definition}
Let $H=(\V, \X, \Init, \Act, \Inv,\Flow,\E,$ $\Pre, \Reset)$ be a HA.  Another hybrid automaton $A\!=\!(\V, \X, \Init_A,\Act, \Inv_A,$ $ \Flow_A, \E, \Pre_A, \Reset_A)$  is a \emph{basic approximation} of $H$ if:
  \begin{itemize}\addtolength{\itemsep}{-4pt}
  \item for all $v \in \V$, $\Inv(v)\Rightarrow{\Inv_A(v)}$, $\Flow(v)\wedge \Inv(v)\Rightarrow{\Flow_A(v)\wedge\Inv_A(v)}$, $\Init(v)\Rightarrow{\Init_A(v)}$;
  \item for every discrete transition  $e \in E$, $\Pre(e)\Rightarrow{\Pre_A(e)}$ and $\Reset(e)\Rightarrow{\Reset_A(e)}$;
  \end{itemize}
where sets of valuations are viewed as predicates.
  If there exists a split  $\theta$ on $H$ such that $A$ is a basic approximation of $H_\theta$ then $A$ is \emph{a phase-portrait approximation} of $H$.
  If the lower and upper bounds of all the predicates in $A$ are rational then $A$ is a (rational) \emph{linear} phase-portrait approximation of $H$.
\end{definition}

A straightforward linear phase-portrait approximation is obtained by replacing the invariant in each mode $v$ by a product of rational intervals that contains $\Inv(v)$ and  all  other predicates, including the flow evolution, by the rational lower and upper bounds implied by the invariant on $v$.


  \begin{example}
  The automaton of Figure~\ref{aut_Thermostat_PA} is the linear phase-portrait approximation of the thermostat  (with the same split as in  Figure~\ref{aut_Thermostat_split}).  Every predicate on $\dot x$ is replaced by a predicate that specifies lower and upper bounds on it.  For example,
  the approximation of $\dot{x}$ in mode ($\mathit{ON}$,1) 
  yields the set $\{\dot{x}\mid3 \leq \dot{x} \leq 4\}$.
  \end{example}

 \begin{figure}

      \centering
 \begin{tikzpicture}[->,>=stealth',shorten >=0.1pt,auto,node distance=3cm,
                    semithick]
  \tikzstyle{every state}=[rectangle, fill=none,draw=black,text=black, minimum width=3.7em]
  \node[state, draw=none] (R) {};
  \node[state,  node distance=14mm] (A)  [below of=R]                   {\begin{minipage}{1.08cm}  \begin{center}\scriptsize{\emph{(\textbf{ON},1)}\\[1pt]  ${3 \! \leq \! \dot{x}\! \leq \! 4}$\\[1pt]  ${1 \!\leq x\! \leq  2}$}\end{center}
\end{minipage}};

 \node[state, node distance=3cm] (D)  [below of=A]                   {\begin{minipage}{1.08cm}  \begin{center} \scriptsize{\emph{(\textbf{ON},2)}\\  ${2 \! \leq \! \dot{x}\! \leq \! 3}$\\ ${2 \!\leq x\! \leq  3}$}\end{center}
\end{minipage}};

  \node[state]         (B) [right  of=A]      {\begin{minipage}{1.29cm}\begin{center}  \scriptsize{\emph{(\textbf{OFF},1)}\\${\!\!-2\! \leq \! \dot{x}\!\leq\!-1\!}$\\ ${1 \!\leq x\! \leq\!  2}$}
\end{center}\end{minipage}};

 \node[state, node distance=3cm]         (E) [below  of=B]      {\begin{minipage}{1.27cm}\begin{center}  \scriptsize{\emph{(\textbf{OFF},2)}\\${\!-3\!\leq\!\dot{x}\!\leq\!-\!2\!}$\\ ${2 \!\leq x\! \leq\!  3}$}
\end{center}\end{minipage}};

  \node at (5,-3)[state]         (C)       {\begin{minipage}{1.2cm}\begin{center}  \scriptsize{\emph{(\textbf{DOWN},1)}\\$\dot{x}\!=\!0$\\ $\!x\!=\!0$}
\end{center}\end{minipage}};
  \path
  (R) edge    node[pos=0.1,right] {\scriptsize{$\!x=\!2$}}(A)
  (D) edge     [right]
  node[pos=0.3, below  ] {\scriptsize{$\!x\!=\!3$}}
  node[pos=0.9, below  ] { 
  }
  node[pos=0.5, above  ] {\scriptsize{\emph{off}}}
  node[pos=0.5, below  ] {}(E)
  (B) edge     [ right]
  node[pos=0.2, below  ] {\scriptsize{$\!x\!=\!1$}}
  node[pos=0.9, below  ] {  
  }
   node[pos=0.5, above  ] {\scriptsize{\emph{on}}}
   node[pos=0.5, below  ] {}(A)
  (B)  edge[in=160,out=0, bend left]
  node[pos=0.25, above ] {\scriptsize{$\!x\!=\!1$}}
     node[pos=1.05, above right] {\scriptsize{$\{0\}$}} node[pos=0.4, below  ]{\scriptsize{\emph{on}}}(C)

   (D) edge    [in=250,out=110]
     node[pos=0.9, below , left ] {
  }
      node[pos=0.5, above , left ] {\scriptsize{\emph{$\tau$}}}  node[pos=0.5, below , left ] {}(A)
   (A) edge      [in=70,out=290]
     node[pos=0.9, below , right ] {
  }
       node[pos=0.5, above,right  ] {\scriptsize{\emph{$\tau$}}}  node[pos=0.5, below  ] {}(D)

   (E) edge    [in=250,out=110]
     node[pos=0.9, below , left ] {
  }
      node[pos=0.5, above , left ] {\scriptsize{\emph{$\tau$}}} node[pos=0.5, below  ] {}(B)

    (B) edge     [in=70,out=290]
       node[pos=0.9, below  ] {
  }
      node[pos=0.5, above,right  ] {\scriptsize{\emph{$\tau$}}} node[pos=0.5, below  ] {}(E)
  ;

     \end{tikzpicture}

   \caption{ A  linear phase-portrait approximation  of the thermostat} \label{aut_Thermostat_PA}
\end{figure}
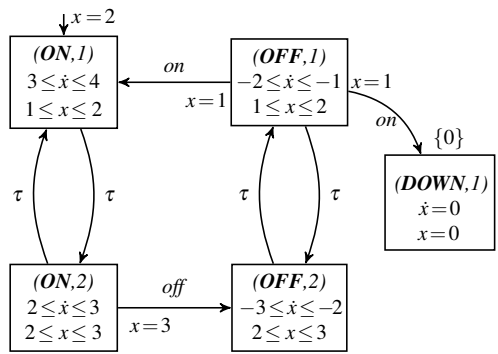

The following theorem implies that if a safety property is verified for an approximation $A$, it holds also for $H$~\cite{Sim90,Sim92}.

  \begin{theorem}\cite{Sim90}
  \label{theorem_approximation}
    If $A$ is a linear phase-portrait approximation of $H$, then $A$ simulates $H$. If it is just a split of $H$ then $A\equiv H$. In both cases, the state  $((v,i),\mathsf{a})$ of $A$ is related to  $(v,\mathsf{a})$ in~$H$.
  \end{theorem}

  The automaton of Figure~\ref{aut_Thermostat_PA} simulates the split of the automaton (Figure~\ref{aut_Thermostat_split}), and then, by transitivity of simulation, simulates the thermostat hybrid automaton.

  The verification of initialized rectangular hybrid automata has been widely discussed, particularly in \cite{Isa06} and \cite{Kop96} and it is proved that their verification  is decidable. Hence, a  non-rectangular hybrid automata can be verified (for satisfaction of safety properties) if an initialized linear phase-portrait approximation can be defined from it.

\section{Analysis of probabilistic hybrid automata}

 In this section, we show how the two methods presented above can be used for probabilistic hybrid automata. Our definition of a probabilistic hybrid automaton (PHA)
\label{def_AHP}
 is close to but slightly more general than the one of Sproston \cite{Spr01}.  We also add the definition of finitely branching PHAs. A (discrete) probability distribution over a set  $C$ is a function $\mu: C\rightarrow{[0,1]}$ such that $\sum_{c\in C}\mu(c)\leq1$; the support of $\mu$ is defined as $\supp(\mu):=\{c \in C\mid \mu(c)>0\}$, and it is countable. 
For $U\subseteq C$, we sometimes write $\mu(U):=\sum_{c\in U}\mu(c)$.    Let $\Dist(C)$ be the set of all (discrete) distributions over $C$.



 \begin{definition}
     \label{defSHP} 
A tuple
$\,H=(\V, \,\X, \,\Init,\, \Act,\, \Inv, \Flow,\prob,\langle \pre_{v,a}\rangle_{v \in \V\!,\, a \in \Act}, {\langle \pos_{v,a}\rangle_{v \in \V\!,\, a \in \Act})}$
 is a \emph{probabilistic hybrid automaton (PHA)} if\,\,
  $\V$, $\X$  $\Init$, $\Act$, $\Inv$ and $\Flow$ are as in Def.~\ref{defHA} and
    \begin{itemize}\addtolength{\itemsep}{-4pt}
       \item $\prob: V\times \Act \rightarrow\cP_{\mathit{fin}}(\mathrm{Dist}(V\times  \cP(\mathbb{R}^X_*)))$
     encodes probabilistic transitions. 
If $\tau\in\Act$, we require that every $\mu\in \prob(v,\tau)$ is concentrated in a unique pair of the form $(v,\{\mathsf{d}\})$.
       \item $\pre_{v,a}: \prob(v,a) \rightarrow\cP(\mathbb{R}^X)$ defines  preconditions for distributions from $v \in \V$ and $a \in \Act$.
       \item $\pos_{v,a}: \prob(v,a)\times V \rightarrow\cP(\mathbb{R}^X)$ defines  postconditions for distributions associated with $v \in \V$ and $a \in \Act$.
     \end{itemize}
We say that $H$ is \emph{finitely branching} if for every $v\in V$, $a\in \Act$, $\mu\in\prob(v,a)$,
$\mu$ is \emph{finitely branching}, that is, $\supp(\mu)$ is finite   \jo{and every set $\post$ such that $(v',\post)\in \supp(\mu)$ is also finite.}
     \end{definition}
To simplify the notation, we drop the subscripts of $\pre$ and $\pos$ when there is no ambiguity.

The semantics of PHAs, is given by probabilistic transition systems.  States are defined in the same way.
  As for non-probabilistic hybrid automata, we distinguish two kinds of transitions in PHAs.  Flow transitions are  the same, but discrete transitions are now probabilistic and hence defined from a state $(v,\mathsf{a})$ to a distribution.  To define transitions, we need some notations on valuations.
  For  $\mathsf{d}\in \mathbb{R}_*^X$,  $\mathsf{a}\in \mathbb{R}^X$,   $\mathsf{post}\subseteq \mathbb{R}_*^X$,
$ \mathsf{A}\subseteq \mathbb{R}^X $ and $x\in X$, let
$$\mathsf{d}[\mathsf{a}](x)\!\!:=\!\!\begin{cases}\!\mathsf{d}(x)  &\!\!\!\!\! \text{if}\ \mathsf{d}(x)\neq*\\ \!\mathsf{a}(x)\ & \!\!\!\!\! \text{if}\ \mathsf{d}(x)=*,\end{cases}
 \quad \mbox{and}\quad
\begin{array}{ll}
\mathsf{post}[\mathsf{A}]\!:=\!\!\{\mathsf{d}[\mathsf{a}]\mid \mathsf{d}\in  \mathsf{post}\!, \mathsf{a}\in \mathsf{A}\}
\\\mathsf{post}(x)\!\!:=\!\!\{\mathsf{d}(x)\mid\mathsf{d}\in \mathsf{post}\}.
\end{array}$$
Transitions of action $a$ from a state $(v,\mathsf{a})$ in the underlying PTS are as follows. Let $\mu \in \prob(v,a)$, $\mathsf{a} \in \pre_{v,a}(\mu)$,  $\supp(\mu)\!=\!\!\{(v_i,\post_i)\}_{i=1}^m $. Each combination of $\d_i\in\post_i$, $i\!=\!1,\ldots,m$, such that $\d_i[\mathsf{a}]\in\pos(\mu,v_i)$, defines a transition
 $$(v,\mathsf{a})\overset{a}\rightarrow \mu_\mathsf{a}^{\langle \d_i\rangle},$$
where $\mu_\mathsf{a}^{\langle \d_i\rangle}$ is positive  on (arrival) states $(v_i,\d_i[\mathsf{a}])$;  the probability that the automaton transits to a state $(v',\mathsf{a}')$ is
\begin{align*}
\mu_\mathsf{a}&^{\langle \d_i\rangle}(v',\mathsf{a}'):=
\begin{cases}
\displaystyle\sum_{{i=1}}^m\,\{\mu(v_i,\post_i)\!\mid\! v_i\!=\!v'\!,\ \! \d_i[\mathsf{a}]\!=\!\mathsf{a}'\}& \mbox{ if }\ \mathsf{a}'\! \in \! \pos(\mu,v')\\ \qquad0 &\mbox{~~otherwise.}\end{cases}
\end{align*}

  \begin{example}
A probabilistic version of the thermostat is shown in Figure~\ref{aut_Thermostat_Prob1}.
     \begin{figure}

  \begin{center}
  \begin{tikzpicture}[->,>=stealth',shorten >=0.3pt,auto,node distance=3cm,
                    semithick]
  \tikzstyle{every state}=[rectangle,fill=none,draw=black,text=black, minimum width=0em]
  \node[state, draw=none] (R) { };
  \node[state,  node distance=1cm] (A)  [right of=R]                   {\begin{minipage}{1.09cm}  \begin{center}  \scriptsize{\emph{\textbf{ON}}\\  ${\!\dot{x}=\!-x\!+\!5\!}$\\ ${1 \!\leq x\! \leq\!3}$} \end{center}
\end{minipage}};

  \node[state]         (B) [right  of=A]      {\begin{minipage}{0.94cm}\begin{center}  \scriptsize{\emph{\textbf{OFF}}\\${\dot{x}\!=\!-x}$\\ ${\!1 \!\leq x\! \leq\!  3}$}
\end{center}\end{minipage}};

  \node[state]         (C) [right  of=B]      {\begin{minipage}{0.9cm}\begin{center}  \scriptsize{\emph{\textbf{DOWN}}\\$\dot{x}=0$\\ $x=0$}
\end{center}\end{minipage}};

  \path   (R) edge  node[pos=-0.7,above] {\scriptsize{$\!x\!=\!2$}}(A)
  (A) edge     [bend right]     node[pos=0.17, below ] {\scriptsize{$\!x\!=\!3$}}  node[pos=0.9, below , sloped ] {
  }    node[pos=0.5, above , sloped ] {\scriptsize{\emph{off}}} node[pos=0.5, below , sloped ] {}
  node[pos=0.6, below , sloped ] {\tiny{$1$}}
  (B)
  (B) edge     [bend right]     node[pos=0.22, below  ] {\scriptsize{$\!x\!=\!1\!$}}  node[pos=0.9, below , sloped ] {
  }
      node[pos=0.5, above , sloped ] {\scriptsize{\emph{on}}} node[pos=0.5, below , sloped ] {}
      node[pos=0.6, below , sloped ] {\tiny{$0.9$}}
      (A)
  (B)  edge[in=160,out=0, bend left]     node[pos=0.17, above  ] {\scriptsize{$x=1$}}         node[pos=0.5, below , sloped ] {} node[pos=0.95, below , sloped ] {\scriptsize{$\{0\}$}} node[pos=0.5, above , sloped ]{\scriptsize{\emph{on}}}
   node[pos=0.6, below , sloped ] {\tiny{$0.1$}}
   (C);
     \end{tikzpicture}
\end{center}
\caption{A probabilistic version of the thermostat}\label{aut_Thermostat_Prob1} 
\end{figure}
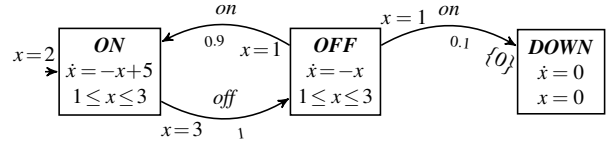
Each discrete transition is labeled by a probability value and an action.
  Since there is only one variable, a valuation can be represented as a real number.     For mode $\OFF$, we have $\prob(\mathit{OFF}, on)=\{\mu\}$ where $\pre(\mu)=\{1\}$,  such that $\mu(\mathit{ON}, \{*\})=0.9$, that is,  the temperature is unchanged,
  and $\mu(\mathit{DOWN},\{0\})=0.1$. 
   Here, defining $\pos(\mu,-)$ is not necessary since it is encoded in $\mu$. Suppose that in another example  one would set $\mu(\mathit{ON}, [1;3])=0.9$. This would mean that the temperature would end up in the interval $[1;3]$ and that the exact value would happen non deterministically.  This example would not be finitely branching.

    \end{example}

 \jo{We now discuss how Definition~\ref{defSHP} of PHA slightly differs from previous ones~\cite{{Spr01}}.}
First note that non deterministic transitions in PHA arise in two ways: when the image of $\prob$ is not a singleton, and from the possible combinations $\d_i\in\post_i$, $i\!=\!1,\ldots,m$, that we can obtain. Hence, the expression \emph{finitely branching} is well chosen because  every set $\post$, such that $(v,\post)\in \mu$, being finite in the underlying probabilistic
transition system, every state $(v,\mathsf{a})$ will have finitely many distributions $\mu_{\mathsf{a}}^{\langle \d_i\rangle}$ associated to any action.

 \jo{The star notation is more general than the reset set of Sproston~\cite{Spr01}  if there is more than one variable. In the latter,  the target of a transition is a distribution over $V\times \cP(\mathbb{R}^n)\times \cP(X)$, where the third component of a tuple $(v,\post,X')$, called a \emph{reset set}, represents the set of variables that can change value in the transition, with respect to valuations of $\post$.  The star notation allows to state, for example, that valuation $(x,y)$ will be modified to $(x,0)$, $(x,2)$, $(0,y)$, $(1,1)$ with probability 1 non deterministically.  The corresponding set would  be $\post_0:=\{(*,0),(*,2),(0,*),(1,1)\}$. This is an important feature to describe the transitions of the clock-translation (see  Section~\ref{sec_prob_clock}) and is not possible with the reset set because there is no uniform reset of any variable in $\post_0$.  Note in passing that the star notation avoids a third component in the notation and will allow to state very simply the notion of initialised PHA. }
The  use of a postcondition function together with the star notation allows to define distributions on complex sets, such as:  $\post_0\cap ([0;3]\times[1;4])$, the distribution being defined on $\post_0$ and the rectangle being the postcondition. We could not transfer the latter into the distributions by defining $\mu$, for example, as
 $\mu(v',\post_0\cap([0;3]\times[1;4]))$ since $\post_0$ may contain valuations that assign $*$ to a variable which implies that its value depends on the actual state, or valuation for which the distribution $\mu$ will be used.  Postconditions are  necessary for the splitting of PHAs, if one wants to avoid  the even more general model that  consists in defining probabilistic transitions  $\prob$ from  $S_H\times A$ to $\cP_{\mathit{fin}}(\mathrm{Dist}(V\times  \cP({\mathbb{R}^X})))$.
This  is too general for practical purposes: indeed, any description of a system must be finite and hence states must be  described in a parametric (or generic) way.  

  The condition on $\tau$ transitions is to simplify the definition of weak transitions, since it permits to write $s\overset{\tau}\rightarrow{s'}$. There may be more than one $\tau$ transition from a mode $v$, but for each $\mu\in \prob(v,\tau)$, there is a unique pair $(v',\{\d\})$ such that $\mu(v',\{\d\})=1$.  Then, following definition of $\mu_\mathsf{a}^{\langle \d_i\rangle}$ above, there is a $\tau$-transition to state $(v',\d[\mathsf{a}])$ if and only if  $\d[\mathsf{a}]\in\pos(\mu,v')$ (that is, $\mu_\mathsf{a}(v',\d[\mathsf{a}])=1$).
   Weak flow transitions are then defined between states as for hybrid automata, and we write $s \stackrel{a}\twoheadrightarrow{}\mu$ if there exists a finite sequence of transitions $s\stackrel{\tau}\rightarrow{s_1}\overset{\tau}\rightarrow{s_2}\ldots\overset{\tau}\rightarrow{s_k}\overset{a}\rightarrow{\mu}$.


We now define a notion of weak simulation between PHAs.
Let ${\preceq\subseteq S\times T}$ be a relation between two sets $S$ and $T$.  For $X\subseteq S$, we use the notation $\preceq(X):=\{t\in  T\mid \exists s\in X. s\preceq t\}$.

\begin{definition}
\label{rel_sim_bis}
Let $H_1$, $H_2$ be two probabilistic hybrid automata.  A relation $\preceq\subseteq S_{H_1}\times S_{H_2}$  is a \emph{simulation} if any initial state of $H_1$ is related to an initial state of $H_2$ and whenever $s_1\preceq s_2$, we have:
\begin{itemize}
\item if $s_1\stackrel{a} {\twoheadrightarrow}\mu_1$, for $a\in\Sigma
\setminus \{\tau\}$ then $s_2\stackrel{a} {\twoheadrightarrow} \mu_2$
and $\mu_1(X) \leq \mu_2(\preceq(X))$ for every $X$;
\item if $s_1\stackrel{\sigma}{\twoheadrightarrow} s'_1$, for $\sigma \in \mathbb{R}_{\geq0}$,
$s_2\stackrel{\sigma }{\twoheadrightarrow} s'_2$ and $s'_1\preceq s'_2$.
\end{itemize}
Then we say that $H_1$ is simulated by $H_2$, written ${H_1\preceq H_2}$.
 If $\preceq^{-1}$ is also a simulation, it is a bisimulation.  Equivalently,
an equivalence relation is a \emph{bisimulation} if in the condition above we have $\mu_1(X) =\mu_2(X)$ for each equivalence class $X$.
  \end{definition}
This definition is known to be equivalent to the one using weight functions: see Desharnais et al.~\cite{DesLavTra08} for a proof that the inequality between $\mu_1$ and $\mu_2$ above  is equivalent to the existence of a network flow between them; it is well-known~\cite{Baier96b}, in turn, that the flow condition is equivalent to the existence of a  weight function between $\mu_1$ and $\mu_2$.



  \bigskip



\subsection{Probabilistic clock-translation}\label{sec_prob_clock}         
 \jo{   In this section, we prove that clock-translation~\cite{Sim90}, can be applied to a PHA and that it results in a  bisimilar PHA, as expected.}
The general method is to first compute a non probabilistic hybrid automaton from a PHA.  Then we apply clock-translation and finally add probabilities in order to obtain a probabilistic clock-translation. There is no condition for computing the underlying non probabilistic HA but we need a notion of solvability so that we will be able to use the clock-translation method thereafter.  A variable  $x$ of a PHA $H$ is \emph{solvable} if \begin{itemize}
\item the two first conditions of solvability for non probabilistic systems are satisfied
\item 
 for every state $s\in S_H$, if $s\overset{a}\to\mu$, and  $\mu(v',\post)>0$, then $\post(x)=\{r\}$ for some $r\in \mathbb{R}_*$.  Moreover,  if $\post(x)=\{*\}$, then we must have $f^v_x=f^{v'}_x$.\smallskip
\end{itemize}
Let $H_p=(\V, \X, \Init, \Act, \Inv, \Flow, \prob,\langle \pre\rangle$, $\langle \pos\rangle)$ be a PHA.
 The algorithm has three steps:   
\\[10pt]
\emph{Step 1}:  Define $H\!\!:=\!\!(\V, \X, \Init, \A, \Inv, \Flow, \E, \Pre,\Reset)$ the \emph{underlying non probabilistic} HA of~$H_p$ as:
       \begin{itemize}
       \item  $\begin{array}[t]{ll}\A:=\{\!\!&a_{\mu}^\post \mid \exists v,v'\ \in \V \mbox{ such that } \mu \in \prob(v,a) \mbox{ and } \mu(v',\post)>0\}\end{array}$.
       \item   $\E\!:=\!\{\!(v,\! a_{\mu}^\post,\!v')\mid\! \mu \! \in\! \prob(v,a)\mbox{ and } \mu(v', \post)>0\}$. Finally, for $e:=(v,\! a_{\mu}^\post,\!v')\in E$, we set \\
       $\Pre(e):= \pre_{v,a}(\mu)$ and  $\Reset(e,\mathsf{a}):=\post[\mathsf{a}]$.
       \end{itemize}
    Note that solvability of $H_p$ implies that $\Reset^x (e)=\{r\}\subseteq\mathbb{R}_*$ and hence $H$ is also solvable.\\[10pt]
 \emph{   Step 2}: {Since $H$ is solvable, let $T=(\V_T, \X_T, \Init_T, \A, \Inv_T, \Flow_T, \E_T, \Pre_T, \Reset_T)$ be the clock-translation of $H$ w.r.t.~$x$. Hence, each transition $e=(v,a_{\mu}^\post,v')$ of $H$ becomes a transition $e_T=((v,c), a_{\mu}^\post, (v',r))$ in $T$, where $r=c$ if $\post(x)=\{*\}$, otherwise  $\post(x)=\{r\}$. }
 \\[10pt]
  \emph{  Step 3}:
  Finally, we build ${T_p=(\V_T, \X_T, \Init_T, \Act, \Inv_T, \Flow_T, \Prob,\langle \Pre \rangle,\langle \Pos \rangle)},$ the probabilistic clock-translation  of $H_p$ from $T$ as follows.
   Let $(v,c)$ be in $\V_T$, and $a$ in $\Act$. $\Prob((v,c),a)$ contains all distributions $\nu_\mu$,  defined from some $\mu \in \prob(v,a)$ as follows:
\begin{itemize}
\item for each edge $e_T=((v,c), a_{\mu}^\post, (v',r))$ such that $\Reset^{t_x}_T(e_T)=\{0\}$ (i.e., $\Reset^x(e)=\{r\}$), let
  \vspace{-1mm}\begin{align*}
\nu_\mu((v',r),\post[t_x\mapsto 0]):=&\mu(v',\post);
\end{align*}  \vspace{-5mm}
\item for each transition $e_T=((v,c), a_{\mu}^\post, (v',c))$ such that $\Reset^{t_x}_T(e_T)=\Reset^x(e)=\{*\} $, let  \vspace{-1mm}
$$\nu_\mu((v',c),\post[t_x\mapsto *]):=\mu(v',\post).
$$  
\end{itemize}
 \jo{For both cases,  $\Pre(\nu_\mu):={{\Pre_{T}}(e_T)}$, and $\Pos(\nu_\mu,(v',r)):=\{\mathsf{a}\in \mathbb{R}_*^{X\cup\{t\}}\mid  \mathsf{a}|_X \in  \pos(\mu,v')\}$.
}

   \begin{example}
    \label{exProbCT}
 \joo{ The clock-translation of the probabilistic thermostat automaton is a slight modification of the HA of Figure~\ref{aut_Thermostat_CT};  all transitions get probability $1$ except for $\mathit{on}$-transitions from $(\mathit{OFF},3)$ which have probability $0.9$ to $(\mathit{ON},1)$ and $0.1$ to $(\mathit{DOWN},1)$. }

    \end{example}

 We should now prove that the construction yields a valid PHA: this will be a consequence of the following theorem.

%


       \begin{theorem}
       \label{theorem_bissim_prob}If $T_p$ is the clock-translation of $H_p$, then $H_p$ and $T_p$ are bisimilar.
       \end{theorem}

       \begin{proof}{}
       \label{preuve_theorem_bisim_prob}
Let $H$ be the underlying non probabilistic automaton of $H_p$ and $T$ its clock-translation.  By Theorem \ref{theorem_bis_trans}, $H$ and $T$ are bisimilar through $\eta:S_T\to S_H$ which, being a function, returns a unique state for any state of $T$.
%
As $H_p$ and $H$ have the same state space, similarly for $T$ and $T_p$, $\eta$ can be seen as a function between states of $T_p$
and $H_p$.   We prove that $\eta$ is a  bisimulation between  $T_p$ and $H_p$.

Let $s=((v,c),\mathsf{a})$ be a state of $T_p$.  There are two kinds of transitions to check in the definition of simulation.
For flow transitions, let $s\overset{\sigma}{\twoheadrightarrow}{s'}$, where $\sigma\in \mathbb{R}_{>0}$. Then  since $\eta$ is a bisimulation between $H$ and $T$, we obtain $\eta(s)\overset{\sigma}{\twoheadrightarrow}\eta(s')$, as wanted.
 \jo{For discrete transitions, we have to prove that for all $\nu_\mu\in \Prob((v,c),a)$ defined from $\mu\in \prob(v,a)$, 
for $\mathsf{a}\in \Pre(\nu_\mu)$ and any combination ${\langle \d_i\rangle}$ from the support of $\nu_\mu$,  we have ${\nu_{\mu,\mathsf{a}}^{\langle \d_i\rangle}(\eta^{-1}(U))=\mu_{\mathsf{a}|_X}^{\langle \d_i|_X\rangle}(U)}$ for all  $U \subseteq S_{H_p}$.} In fact, we need only to prove it for $U$ equals to some state $(v',\mathsf{a'})$ since $\{(v',\mathsf{a'})\}\cup \eta^{-1}((v',\mathsf{a'}))$ is an equivalence class.
\begin{align*}
\nu_{\mu,\mathsf{a}}^{\langle \d_i\rangle}( \eta^{-1}((v',\mathsf{a'})))
&=  \sum_{\mathsf{b},r}\,\{\, \nu_{\mu,\mathsf{a}}^{\langle \d_i\rangle} ((v',r),\mathsf{b})\mid\underbrace{\mathsf{b}|_X=\mathsf{a}', \ g_r(\mathsf{b}(t_x))=\mathsf{a}'(x)}_{P(\mathsf{b},u)}\}\\
&= \!\sum_{\mathsf{b},r}\, \sum_{{  i=1}}^m \,
     \{\,\nu_\mu((v_i,r_i),\post_i)\mid \begin{array}{ll}v_i=v', r_i=r,\ \d_i[\mathsf{a}]=\mathsf{b}, \\  \post_i(x)=\{r_i\}, P(\mathsf{b},r_i)\end{array}\}\\
&\qquad  \cup \{\,\nu_\mu((v_i,r_i),\post_i)\mid \begin{array}{ll}v_i=v', r_i\!=\!r\!=\!c,\ \d_i[\mathsf{a}]=\mathsf{b}, \\  \post_i(x)=\{*\}, P(\mathsf{b},c)\end{array}\}\\
&=\! \sum_{{  i=1}}^m\,\{\,\nu_\mu((v_i,r_i),\post_i)\mid \begin{array}{ll}v_i=v',\ \d_i[\mathsf{a}]|_X=\mathsf{a}',\\ \post_i(x)=\{r_i\}, \  \mathsf{b}(t_x)=0\end{array}\}\\
&\qquad  \cup \{\,\nu_\mu((v_i,c),\post_i)\mid \begin{array}{ll}v_i=v',\ \d_i[\mathsf{a}]|_X=\mathsf{a}', \\ \post_i(x)=\{*\}, \  \mathsf{b}(t_x)=\mathsf{a}(t_x)\end{array}\}\\
&=\!\sum_{{  i=1}}^m\,\{\,\mu(v_i,\post_i)\mid \begin{array}{ll}v_i=v',\ \d_i[\mathsf{a}]|_X=\mathsf{a}', \\\post_i(x)=\{r_i\}\mbox{ or }\{*\}\end{array}\}
\\&=\;\mu_{\mathsf{a}|_X}^{\langle \d_i|_X\rangle}(v',\mathsf{a'}).
\end{align*}
In the third equality, the double sum is reduced to a single one because $i$ determines $\mathsf{b}$ and $r$.    
       \end{proof}

   \begin{corollary}
      The clock-translation of  a solvable PHA is a PHA.
   \end{corollary}

    \begin{proof}{}
   \label{preuv_valid_p}
We only need to prove that every defined $\nu_\mu$ is a distribution, that is,  $\nu_\mu(S_{T_p})$ is 1.  We do so by showing that elements of $\supp(\nu_\mu)$ are in bijection with elements of $\supp(\mu)$.   By construction, for every $\d\in\post$ such that $\nu_\mu(s,\post)>0$,  we have $\d(t_x)=\{0\}$ if and only if  $\d(x)=r$ and $\d(t_x)=*$ if and only if  $\d(x)=*$.  This implies that $\d|_X=\d'|_X$ if and only if $\d=\d'$, as wanted.  Another proof is obtained by taking $U=S_{T_p}$ in the proof of Theorem~\ref{theorem_bissim_prob}.
\mbox{}
\end{proof}

    If all variables of $H_p$ are solvable, then its clock-translation with respect to all its variables will yield a probabilistic timed automaton.  Knowing that the model-checking of
       probabilistic timed automata is decidable~\cite{Spr02}, it implies that  the model-checking of  solvable PHAs is decidable.



\subsection{Probabilistic linear phase-portrait approximation}         

   In this section, we show how to apply  the linear phase-approxima\-tion method to a probabilistic hybrid automaton, and that it results in a rectangular hybrid automaton which simulates it.

Let $H_p\!=\!(\V, \X, \Init, \Act, \Inv, \Flow, \prob,\langle \pre\rangle, \langle \pos \rangle)$ be a fi\-ni\-tely branching PHA. The method of approximation in a probabilistic context follows the same kind of steps as the clock-transla\-tion, by going through an underlying non probabilistic hybrid automaton (this approach is also the one of Zhang et al.~\cite{cav2010}).  However, this translation is simpler, as well as smaller, here as no condition has to be satisfied by the non probabilistic automaton.
\\[5pt]
\emph{Step 1}:  We define, $H=(\V, \X, \Init, \A, \Inv, \Flow, \E, $ $\Pre, \Reset)$, the \emph{underlying non probabilistic} hybrid
automaton of $H_p$ as follows:

       \begin{itemize}%
       \item  $\A:=\{a_{\mu}\mid \exists v\in \V$ such that $\mu \in \prob(v,a)\}$.
       \item For each $\mu \in \prob(v,a)$:
           $\E$ will contain all ${e=(v, a_{\mu}, v')}$ such that there is some $(v', \post)\in \supp(\mu)$.
$\Pre(e)= \pre(\mu)$.
$\Reset(e,\mathsf{a})=\displaystyle\bigcup_ {(v', \post)\in \supp(\mu)}\!\!
\post[\mathsf{a}]\cap \pos(\mu,v').$
       \end{itemize}
 \emph{   Step 2}: {We build the linear phase-portrait approximation of $H$:
  $$T=(\V_\theta, \X_\theta, \Init_\theta, \A, \Inv_\theta, \Flow_\theta, \E_\theta, \Pre_\theta, \Reset_\theta).$$
  \emph{Step 3}: Finally, we build $$T':=(\V_\theta, \X_\theta, \Init_\theta, \Act, \Inv_\theta, \Flow_\theta,\Prob, \langle \Pre \rangle, \langle \Pos\rangle)$$ the probabilistic linear phase-portrait approximation of $H_p$ from $T$ as follows.
\\[5pt]
 Let $(v,i)$ be in $\V_\theta$, and $a$ in $\Act$. $\Prob((v,i),a)$ contains two kinds of distributions:
 \begin{itemize}
\item $\mu_\theta$  defined as follows, for each $\mu \in \prob(v,a)$.  For each $e_\theta=((v,i), a_{\mu}, (v',j))$ $\in E_\theta$ 
and  $\post$ such that $(v',\post)\in \supp(\mu)$,  
we define
$$\mu_\theta((v',j),\post):=\mu(v',\post),$$
and $\mu_\theta$ is zero elsewhere.  Preconditions and postconditions are
\begin{itemize}
\item ${\Pre(\mu_\theta):=\Pre_\theta(e_\theta)} \cap \inv_i^{v}$ and
\item $\Pos(\mu_\theta,(v',j)):=\overline{\inv}_j^{v'}\cap\Pos(\mu,v')$,
\end{itemize}
where  $\overline{\inv}^{v'}_j:=\inv^{v'}_j\setminus(\cup_{k<j} {\inv}^{v'}_k)$ defines a partition of $\Inv(v')$.
\item if $a=\tau$,  all $\mu_\tau^{j}$ defined as
$$\mu_\tau^{j}((v,j), \{*\}^X):=1,$$
that is, the valuation is unchanged during the silent transition from a copy of $v$ to another. These transitions correspond to the edges $((v,i), \tau, (v,j))\!\in \!E_\theta$.  There is no special precondition or postcondition, and hence we set $\Pre(\mu_\tau^{j}):=\inv_i^v$ and $\Pos(\mu_\tau^{j},(v',j)):=\inv_j^{v'}$.

\end{itemize}

\begin{remark}
 \joo{The use of postconditions, in presence of the star notation, is crucial in the apparently simple definition of $\mu_\theta$ above, both to make it correct, and to indeed permit a simple and clean formulation.
This is on one hand because of the reasons mentioned after Definition~\ref{defSHP}.  On another hand, if we used the less general syntax involving reset sets instead of the star notation in distributions, the preconditions of $\mu_\theta$ would have to take into account the invariant of the arrival state: if a probability is assigned to a pair that ends up not being valid because of the splitting of transitions, the probability ends up missing, i.e., $\mu_\theta$ would not sum up to 1: the machinery to overcome this loss of probability would complicate a lot the notation.}  The use of $\overline{\inv}_j^{v'}$ is also crucial in the definition:  it makes sure that we do not use, in $\mu_\theta$, a probability value from $\mu$ more than once.  Indeed, some states get duplicated in the split (if some $\mathsf{a}$ belongs to more than one element of $\theta(v)$), which has no negative impact in a non probabilistic HA, since duplicated transitions define bisimilar systems. However, in the probabilistic case, we must make sure that we do not give to both copies the probability value that was meant for one copy: by duplicating the value, we would lose the correspondence between copies and the original mode and we could also end up with a weight of more than one: this would result in a function that is not a distribution.  We chose to put the probability on one of the duplicates because silent transitions make sure that the behavior is  preserved.  We could also have chosen to spread the probability to all duplicates uniformly.
\end{remark}

    \begin{example}
    \joo{ A linear phase-portrait approximation of the probabilistic thermostat automaton is obtained by slightly modifying the HA of Figure~\ref{aut_Thermostat_split} and is illustrated in Figure~\ref{aut_Thermostat_Prob_PA}.  All transitions get probability $1$ except for $\mathit{on}$-transitions from $(\mathit{OFF},1)$ which have probability $0.9$ to $(\mathit{ON},1)$ and $0.1$ to $(\mathit{DOWN},1)$.  Had we not restricted the postcondition of $\mu_\theta((\mathit{ON},j),-)$ to $\overline{\inv}_j^{\mathit{ON}}$, it would give probability one to $(\mathit{ON},2)$ and hence the distribution $\mu_\theta\in \prob((\mathit{OFF},1),\mathit{on})$ would sum up to $2$.}
\begin{figure}

  \begin{center}
 \noindent \!\!\!\!\!\!\begin{tikzpicture}[->,>=stealth',shorten >=0.1pt,auto,node distance=3cm,
                    semithick]
  \tikzstyle{every state}=[rectangle, fill=none,draw=black,text=black]
  \node[state, draw=none] (R) {};
  \node[state,  node distance=14mm] (A)  [right of=R]                   {\begin{minipage}{1.09cm}  \begin{center}\scriptsize{\emph{(\textbf{ON},1)}\\[1pt]  ${3 \! \leq \! \dot{x}\! \leq \! 4}$\\[1pt]  ${1 \!\leq x\! \leq  2}$}\end{center}
\end{minipage}};

 \node[state, node distance=3cm] (D)  [below of=A]                   {\begin{minipage}{1.09cm}  \begin{center} \scriptsize{\emph{(\textbf{ON},2)}\\  ${2 \! \leq \! \dot{x}\! \leq \! 3}$\\ ${2 \!\leq x\! \leq  3}$}\end{center}
\end{minipage}};

  \node[state]         (B) [right  of=A]      {\begin{minipage}{1.28cm}\begin{center}  \scriptsize{\emph{(\textbf{OFF},1)}\\${\!-2\! \leq \! \dot{x}\!\leq\!-1\!}$\\ ${1 \!\leq x\! \leq\!  2}$}
\end{center}\end{minipage}};

 \node[state, node distance=3cm]         (E) [below  of=B]      {\begin{minipage}{1.28cm}\begin{center}  \scriptsize{\emph{(\textbf{OFF},2)}\\${\!-3\!\leq\!\dot{x}\!\leq\!-\!2\!}$\\ ${2 \!\leq x\! \leq\!  3}$}
\end{center}\end{minipage}};

  \node[state]         (C) [right  of=B]      {\begin{minipage}{1.23cm}\begin{center}  \scriptsize{\emph{(\textbf{DOWN},1)}\\$\dot{x}\!=\!0$\\ $\!x\!=\!0$}
\end{center}\end{minipage}};

  \path
  (R) edge    node[pos=-0.5,above] {\scriptsize{$\!x=\!2$}}(A)
  (D) edge     [right]
  node[pos=0.2, below , sloped ] {\scriptsize{$\!x\!=\!3$}}
  node[pos=0.9, below , sloped ] { 
  }
  node[pos=0.5, above , sloped ] {\scriptsize{\emph{off}}}
  node[pos=0.6, below , sloped ] {\tiny{1}}(E)
  (B) edge     [ right]
  node[pos=0.2, below , sloped ] {\scriptsize{$\!x\!=\!1$}}
  node[pos=0.9, below , sloped ] {  
  }
   node[pos=0.5, above , sloped ] {\scriptsize{\emph{on}}}
   node[pos=0.6, below , sloped ] {\tiny{0.9}}(A)
  (B)  edge[in=160,out=0, bend left]
  node[pos=0.2, below ] {\scriptsize{$\!x\!=\!1$}}
     node[pos=0.6, below , sloped ] {\tiny{0.1}} node[pos=0.95, below , sloped ] {\scriptsize{$\{0\}$}} node[pos=0.5, above , sloped ]{\scriptsize{\emph{on}}}(C)

   (D) edge    [in=250,out=110]
     node[pos=0.9, below , left ] {
  }
      node[pos=0.5, above , left ] {\scriptsize{\emph{$\tau$}}}  node[pos=0.5, below , left ] {}(A)
   (A) edge      [in=70,out=290]
     node[pos=0.9, below , right ] {
  }
       node[pos=0.5, above,right  ] {\scriptsize{\emph{$\tau$}}}  node[pos=0.5, below , sloped ] {}(D)

   (E) edge    [in=250,out=110]
       node[pos=0.9, below , left ] {
  }
      node[pos=0.5, above , left ] {\scriptsize{\emph{$\tau$}}} node[pos=0.5, below , sloped ] {}(B)

    (B) edge     [in=70,out=290]
       node[pos=0.9, below , sloped ] {
  }
      node[pos=0.5, above,right  ] {\scriptsize{\emph{$\tau$}}} node[pos=0.5, below , sloped ] {}(E)
  ;

     \end{tikzpicture}
\end{center}
  \caption{A phase-portrait approximation of the probabilistic thermostat}\label{aut_Thermostat_Prob_PA}
\end{figure}
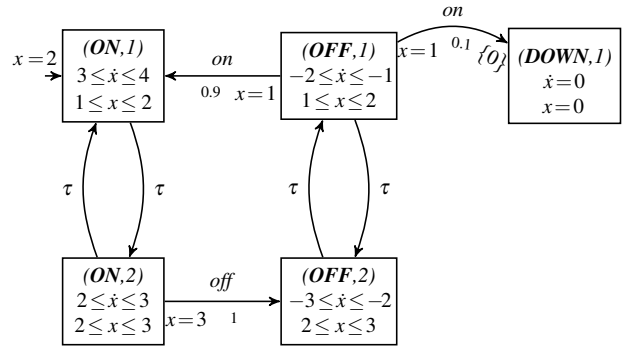
    \end{example}

    }

 We should now prove that the construction yields a valid PHA: this will be a consequence of the following theorem.
      \begin{theorem}
       \label{theorem_simulation_appr} Any linear phase-approximation $H_\theta$ of a
       probabilistic PHA $H_p$ simulates it: i.e., $H_\theta\preceq H_p$.        The split of a  PHA is bisimilar to it.
       \end{theorem}
      \begin{proof}{}
Let $R$ be the relation that relates any $(v,\mathsf{a})\in S_{H_p}$ to every $((v,i),\mathsf{a})$, with  $1\leq i\leq|\theta(v)|$.  Let $(v,\mathsf{a})\in S_{H_p}$, $a\in\Act$, $1\leq i\leq|\theta(v)|$ and $\mu\in \prob(v,a)$.  We have to prove that for all  $\mathsf{a}\in \Pre(\mu_\theta)$ and any combination ${\langle \d_j\rangle}$ from  $\mu_\theta$,  we have $\mu_\mathsf{a}^{\langle \d_j\rangle}(U)\leq \mu_{\theta\mathsf{a}}^{\langle \d_j\rangle}(R(U))$  for all  $U \subseteq S_{H_p}$.  In fact, we show equality. This does not give us a bisimulation because of the flow transitions which only satisfy simulation. The  second equality below relies on the $\overline{\inv}_k^{v'}$'s being disjoint and cumulating to $\inv(v)$:  indeed, for each $k$, there is only one $\mathsf{a}'\in \overline{\inv}_k^{v'}$ and all $\mathsf{a}\in\inv(v)$ is in one of those.  Consequently the summation over $k$ can be inserted freely.
\begin{eqnarray*}
 \mu_\mathsf{a}^{\langle \d_j\rangle}(U)&=&\!\!\!\!\!\!\sum_{\substack{(v',\mathsf{a}')\in  U}}
\sum_{{  j=1}}^m\,\{\,\mu(v_j,\post_j)\mid v_j=v',\ \d_j[\mathsf{a}]=\mathsf{a}'\} \ where\ m=|\supp(\mu)|\\[-2pt]
&=&\!\!\!\!\!\!\!\!\sum_{\substack{(v',\mathsf{a}')\in  U}}\,\sum^{|\theta(v')|}_{\substack{k=1\\[-2pt]\mathsf{a}'\in\overline{\inv}_k^{v'}}  }\,\sum_{{  j=1}}^m\,
\{\,\mu(v_j,\post_j)\mid v_i=v',\ \d_j[\mathsf{a}]=\mathsf{a}'\}\\[-2pt]
&=&\!\!\!\!\!\!\!\!\sum_{\substack{(v',\mathsf{a}')\in  U}}\,\!\!\sum^{|\theta(v')|}_{\substack{k=1\\[-2pt]\mathsf{a}'\in{\inv}_k^{v'}}  }\,\!\!\sum_{{  j=1}}^m\,
\{\,\mu_\theta((v_j,k),\post_j)\mid\! v_i=v',\! \d_j[\mathsf{a}]=\mathsf{a}'\}\\[-2pt]
&=&\!\!\!\!\sum_{\substack{(v',\mathsf{a}')\in  U}}\!\!\sum^{|\theta(v')|}_{\substack{k=1\\\mathsf{a}'\in\inv_k^{v'}}  }\!\!\mu_{\theta\mathsf{a}}^{\langle \d_k\rangle} ((v',k),\mathsf{a}')
\;=\;
\mu_{\theta\mathsf{a}}^{\langle \d_j\rangle}(R(U))
\end{eqnarray*}
This completes the proof of the first claim.
The second claim follows easily.     \end{proof}

      \begin{example}
       The  linear phase-portrait approximation of the probabilistic thermostat automaton of Figure~\ref{aut_Thermostat_Prob_PA} simulates the thermostat automaton of Figure~\ref{aut_Thermostat_Prob1}. In particular, consider the states $s_1=(\mathit{OFF}, 1)$ and $s_2=(\mathit{ON},1)$ of the original thermostat automaton, such that $\mu(\mathit{ON}, 1)=0.9$ for $\mu \in \prob(\mathit{OFF}, \mathit{on})$. The states (($\mathit{OFF}$,1), 1) and (($\mathit{ON}$,1),1) simulate respectively $s_1$ and $s_2$ since $\mu_\theta((\mathit{ON},1), 1)=0.9$ for $\mu_\theta \in \Prob((\mathit{OFF},1),\mathit{on})$.
      \end{example}

   \begin{corollary}
Phase-portrait approximation and splitting of  fi\-nitely branching PHAs are~PHAs.
   \end{corollary}

   \begin{proof}{}
   \label{preuv_valid_p2}
We only need to prove that every defined $\mu_{\theta\mathsf{a}}$ is a distribution, that is, the total probability out of $\mu_{\theta\mathsf{a}}$ is 1.   This is guaranteed by $\mu_{\mathsf{a}}$ being a distribution and by taking  $U:=S_\theta$ in  the proof of Theorem~\ref{theorem_simulation_appr}; one obtains that $\mu_{\mathsf{a}}(S_H)=\mu_{\theta\mathsf{a}}(S_\theta).$  Since the distribution has nothing to do with the flow evolution, it is the same argument for both cases of approximation and splitting.
  \end{proof}

      Since an approximation of a probabilistic hybrid automaton
       is a rectangular PHA, its model-checking is decidable \cite{Spr01}. Therefore, by taking the right approximation to it, any probabilistic hybrid automaton can be verified.

\section{Conclusion}

     In this paper, we proved that the two methods of Henzinger et al.~\cite{Sim90}, clock-translation and linear phase-portrait approximation, can also be applied in the probabilistic context to verify non-rectangular PHAs.   \joo{The adaptation of the methods to PHAs were facilitated by a modification of the syntax over PHAs: a star notation to represent stability of a variable after a transition and postcondition functions.  The advantage of adapting the methods instead of defining them from scratch is mainly that proving bi/simulation had only to be checked for discrete  transitions.}  The correctness of the constructions is ensured by bi/simulation relations.

     The first method, when it is applicable, results in a probabilistic timed automaton which satisfies exactly the same properties as the original PHA. The inconvenience of this method is that it requires, as in the non-probabilistic case, that  the non-rectangular variables be solvable: all the equations induced by the flow evolution have solutions in $\mathbb{R}$. 

     For linear phase-portrait approximation, there is no restriction on the PHA, and its application results in a rectangular PHA that simulates the original one.
     When a safety property is satisfied by the rectangular approximation, we can assert that the original PHA satisfies the same property. However, in the case when  a safety property is not satisfied by the approximation, more splits should be done; this could be costly in time depending on the property and the size of the PHA.

Side contributions of this paper are also: new additions to the definition of PHA that make them more general;
 a splitting construction on PHAs that results in a bisimilar PHA;
 a simpler description of the two  techniques than what can be found in the original papers~\cite{Sim90,Sim92}.  About the additions to the definition of PHAs, it is interesting to note that the star notation and the postcondition function on distributions permit to represent constraints on variables in terms of set of valuations instead of predicates.  When working with distributions, set of valuations are more natural than predicates.

Let us discuss how our approximation relates to the construction of a recent paper by Zhang et al~\cite{cav2010}, which also defines approximations for probabilistic hybrid automata.  That paper also gives a method to over approximate the original automaton.  It is clear that the latter is less abstract than the former since Zhang et al.\ define a \emph{finite} approximation.  As in the construction of Henzinger et al.~\cite{Sim90}, they start from a cover of the state space and abstract according to it.  The difference is that the cover is over states instead of over the variables' space of values.  More importantly, whereas we use the cover to ``linearize" each piece that the cover defines, they group together all these states into one single state.  The result is a finite probabilistic transition system, whereas we obtain a PHA. They prove, as we do, that their abstraction simulates the original PHA.  However, their approximation is more abstract, which has the advantage of being smaller but of course with less information.

Future work includes the implementations of the technique into a probabilistic model checker and taking advantage of other approximation techniques that have been developed for probabilistic systems in order to widen the class of PHAs for which model-checking is supported.

\subsection*{Acknowledgement}
J. Desharnais wishes to thank Marta Kwiatkowska and the Computing Laboratory of Oxford University for welcoming her during year 2009-2010.

\bibliographystyle{plain}

\end{document}